\documentclass[a4paper,11pt]{article}

\usepackage{amsmath}
\usepackage{amssymb}
\usepackage{amsthm}
\usepackage{array}
\usepackage{amscd}
\usepackage{cases}
\usepackage{bm}
\usepackage[all]{xy}
\usepackage{authblk}
\usepackage{cancel}
\usepackage{ulem}
\usepackage{mathrsfs}

\theoremstyle{plain}
\newtheorem{thm}{Theorem}[section]
\newtheorem{prop}[thm]{Proposition}

\theoremstyle{definition}

\theoremstyle{definition}
\newtheorem{rem}{Remark}

\theoremstyle{remark}

\numberwithin{equation}{section}

\title{A $q$-analogue of the matrix fifth Painlev\'e system}
\author{Hiroshi KAWAKAMI\thanks{\texttt{kawakami@math.aoyama.ac.jp}}}
\affil{Graduate School of Mathematical Sciences, The University of Tokyo,
3-8-1 Komaba, Meguro-ku, Tokyo 153-8914, Japan.}
\date{}

\begin{document}
\maketitle

\begin{abstract}
We consider a degeneration of the $q$-matrix sixth Painlev\'e system.
As a result, we obtain a system of non-linear $q$-difference equations,
which describes a deformation of a certain non-Fuchsian linear $q$-difference system.
We define the spectral type for non-Fuchsian $q$-difference systems and
characterize the associated linear problem in terms of the spectral type.
We also consider a continuous limit of the non-linear $q$-difference system
and show that the resulting system of non-linear differential equations coincides with the matrix fifth Painlev\'e system.
\end{abstract}


\section{Introduction}\label{sec:intro}

The Painlev\'e equations are non-linear second order ordinary differential equations that define novel transcendental functions.
Historically, the Painlev\'e equations were classified into six equations.
We refer to them as $P_{\mathrm{I}}, P_{\mathrm{II}}, \ldots,P_{\mathrm{VI}}$.
The sixth Painlev\'e equation $P_{\mathrm{VI}}$ serves as the ``source'' from which all the other Painlev\'e equations can be derived through degeneration processes.

Since the 1990s, various generalizations of the Painlev\'e equations have been proposed in the literature,
such as discretizations, higher-dimensional analogues, quantizations, and so on.

Recently, Painlev\'e-type differential equations with four-dimensional phase space have been classified
from the perspective of isomonodromic deformations of linear differential equations~\cite{HKNS, K3, K4, K5}.
This series of studies shows that,
in the four-dimensional case, there exist four ``sources'' as extensions of the sixth Painlev\'e equation.
Namely, they are
\begin{itemize}
\item the Garnier system~\cite{G}, which is a classically known multivariate extension of $P_{\mathrm{VI}}$,

\item the Fuji-Suzuki-Tsuda system~\cite{FS1, Ts2014}, which is an extension of $P_{\mathrm{VI}}$ with the affine Weyl group symmetry of type $A$,

\item the Sasano system~\cite{Ss}, which is an extension of $P_{\mathrm{VI}}$ with the affine Weyl group symmetry of type $D$,

\item the matrix sixth Painlev\'e system~\cite{B2, HKNS}, which is a non-abelian extension
of $P_{\mathrm{VI}}$.
\end{itemize}
Note that each of the four equations has its extensions defined in arbitrary even dimensions.
These four families are expected to have an impact on fields such as integrable systems, special functions, and so on.

On the other hand,
Sakai~\cite{Sak2001} established an algebro-geometric theory which provides a comprehensive understanding of two-dimensional (or second order) Painlev\'e equations.
According to Sakai's theory, when the phase spaces are two-dimensional, 
the discrete Painlev\'e equations are more fundamental.
Roughly speaking, by classifying a certain kind of rational surfaces, 22 different surfaces are obtained.
From the discrete symmetry of each surface, a discrete dynamical system (a system of difference equations) is generated.
The theory classifies these discrete Painlev\'e equations into three types:
additive difference, multiplicative difference ($q$-difference), and elliptic difference equations.
The Painlev\'e (differential) equations are understood through the continuous limit of these discrete Painlev\'e equations.
In this sense, we can say that the discrete Painlev\'e equations are more fundamental than the Painlev\'e differential equations.

Our aim is, inspired by the two-dimensional case,
to construct a unified framework for discrete Painlev\'e-type equations in higher dimensions.
However, it is difficult to classify algebraic varieties when the phase spaces have four or more dimensions.
Instead, from the standpoint of the classification theory (by Katz~\cite{Katz1995} and Oshima~\cite{Os}) of linear differential equations
and the isomonodromic/connection-preserving deformation theory,
we would like to develop a framework for higher-dimensional Painlev\'e-type equations that involves discrete Painlev\'e-type equations.

In \cite{K7}, we have defined an equivalence relation between spectral types of linear differential equations
(that is, those can be transformed into each other by M\"obius transformations, the Harnad dual, or certain scalar gauge transformations are equivalent)
and shown that there is a tree structure among the equivalence classes including differential equations without continuous deformations.
The tree structure of equivalence classes
corresponds to the additive surfaces in Sakai's list.
As a next step, we investigate multiplicative difference Painlev\'e-type equations in higher dimensions.

Among the four families mentioned above,
$q$-analogues of the Garnier systems, the Fuji-Suzuki-Tsuda systems, and the Sasano systems
have been obtained and studied by several authors~\cite{Sak2005, Sz, Ts2010, Ma}.
Recently, a $q$-analogue of the matrix sixth Painlev\'e system, which we call the $q$-matrix $P_{\mathrm{VI}}$, has been obtained~\cite{K6}.

In this paper we investigate a degeneration of the $q$-matrix $P_{\mathrm{VI}}$
with the aim of constructing a degeneration scheme for higher dimensional $q$-difference Painlev\'e-type equations.
As a result, a system of non-linear $q$-difference equations is obtained,
which corresponds to the matrix fifth Painlev\'e system in the continuous limit.
We tentatively denote the non-linear system by the $q$-matrix $P_{\mathrm{V}}$.
The $q$-matrix $P_{\mathrm{V}}$ is expressed as a deformation equation of a non-Fuchsian linear $q$-difference equation.

This paper is organized as follows.
In Section~\ref{sec:LqDE} we describe how to construct formal solutions to linear $q$-difference systems.
We also give the definition of spectral types for non-Fuchsian linear $q$-difference systems.
In Section~\ref{sec:q-matPVI} we review the $q$-matrix $P_{\mathrm{VI}}$.
In Section~\ref{sec:degeneration} we consider a degeneration of the $q$-matrix $P_{\mathrm{VI}}$.
In Section~\ref{sec:CL} we consider a continuous limit of the $q$-matrix $P_{\mathrm{V}}$ obtained in Section~\ref{sec:degeneration}.
The appendix is devoted to a brief description of the matrix fifth Painlev\'e system.

\bigskip 

\noindent 
\textbf{Acknowledgements} 

\noindent 
This work was supported by JSPS KAKENHI Grant Number JP20K03705
and the Research Institute for Mathematical Sciences, an International Joint Usage/Research Center located in Kyoto University.

\section{Linear $q$-difference systems}\label{sec:LqDE}
In this section, we collect some facts about linear $q$-difference systems.
The formal normal form is used
to define the spectral type for non-Fuchsian systems.

\subsection{Formal normal form: Fuchsian case}\label{sec:formal_reduction_Fuchsian}

Let $q$ be a complex number satisfying $0 < |q| < 1$.
Consider a linear $q$-difference system with polynomial coefficients
\begin{equation}\label{eq:Fuchsian_system}
Y(qx)=A(x)Y(x), \quad A(x)=A_0+A_1x+\cdots+A_Nx^N
\end{equation}
where $A_j \in M_m(\mathbb{C})$ and $A_0, A_N \ne O$.
If $A_0$ and $A_N$ are both invertible, then the system \eqref{eq:Fuchsian_system} is said to be \textit{Fuchsian}.
For simplicity we assume that $A_0$ and $A_N$ are diagonalizable.
In this subsection we outline the procedure to transform the given Fuchsian system into its formal normal forms
at $x=0$ and $x=\infty$.

We use the following well-known fact from linear algebra.
We denote the set of all eigenvalues of a matrix $A$ by $\mathrm{Sp}(A)$.
\begin{prop}\label{thm:Sylvester_eq}
Let $A \in M_m(\mathbb{C})$ and $B \in M_n(\mathbb{C})$.
Then the linear map
\begin{equation}
\varphi : M_{m,n}(\mathbb{C}) \to M_{m,n}(\mathbb{C}), \quad
\varphi(X)=AX-XB
\end{equation}
is an isomorphism of vector spaces if and only if $\mathrm{Sp}(A) \cap \mathrm{Sp}(B)=\emptyset$.
\end{prop}

First we explain the formal normal form at $x=0$.
For the convenience of later discussion, we consider a system of the following form:
\begin{equation}\label{eq:original_system}
Y(qx)=A(x)Y(x), \quad
A(x)=x^r(A_0+A_1x+A_2x^2+\cdots)
\end{equation}
where $A(x)$ is an $m \times m$ formal power series.
Here $A_0$ is invertible and diagonalizable.
$r$ is a non-negative integer.
Let the eigenvalues of $A_0$ be $\theta_1, \ldots, \theta_m$.
We also assume that the system is \textit{non-resonant},
that is, for any $i, j$
\begin{equation}
\theta_j / \theta_i \notin q^{\mathbb{Z}_{\ge 1}}=\{ q^n \mid n \in \mathbb{Z}_{\ge 1} \}.
\end{equation}

Let $P(x)=\sum_{n=0}^\infty P_nx^n$ be an $m \times m$ formal power series with $P_0=I$.
The substitution $Y(x)=P(x)Z(x)$ yields
\begin{equation}
Z(qx)=P(qx)^{-1}A(x)P(x)Z(x).
\end{equation}
We can choose the matrix $P(x)$ so that
\begin{equation}\label{eq:FNF_Fuchs0}
P(qx)^{-1}A(x)P(x)=x^rA_0.
\end{equation}
The matrix $P(x)$ can be constructed as follows.
The equation \eqref{eq:FNF_Fuchs0} can be written as
\begin{equation}
(A_0+A_1x+A_2x^2+\cdots)(I+P_1x+P_2x^2+\cdots)=(I+qP_1x+q^2P_2x^2+\cdots)A_0.
\end{equation}
Equating the coefficients of $x^n \ (n \ge 1)$ on both sides, we obtain
\begin{equation}\label{eq:determine_Pn}
A_0P_n-P_n(q^nA_0)=-\sum_{k=0}^{n-1}A_{n-k}P_k.
\end{equation}
If the coefficient matrices $P_1, \ldots, P_{n-1}$ are determined,
then the equation \eqref{eq:determine_Pn} uniquely determines $P_n$ by Proposition~\ref{thm:Sylvester_eq} and non-resonant condition.
In this way, the matrix $P(x)$ is constructed inductively.
Then the matrix $x^rA_0$ is the formal normal form of \eqref{eq:original_system} in this case.

The construction at $x=\infty$ is almost the same.
Assume that
\begin{equation}
A(x)=x^r(A_0+A_1x^{-1}+A_2x^{-2}+\cdots)
\end{equation}
where $A_0$ is invertible, $r \in \mathbb{Z}_{\ge 0}$.
We can construct the transformation matrix $P(x)=\sum_{n=0}^{\infty}P_nx^{-n}$ at $x=\infty$ such that
$P(qx)^{-1}A(x)P(x)=x^r A_0$ holds in a similar way.

\subsection{Formal normal form: non-Fuchsian case}

In the case that $A_0$ or $A_N$ of \eqref{eq:Fuchsian_system} is not invertible, the construction of the formal normal form can be modified as follows.
Consider at $x=0$
\begin{equation}\label{eq:original_system_non-Fuchsian}
Y(qx)=A(x)Y(x), \quad
A(x)=x^r(A_0+A_1x+A_2x^2+\cdots).
\end{equation}
Here we assume that
\begin{equation}\label{eq:non-Fuchsian_A0}
A_0=\mathrm{diag}(\theta_1, \ldots, \theta_s, 0, \dots, 0)=\Theta \oplus O_{m-s},
\end{equation}
where for any $i, j$
\begin{equation}\label{eq:non-resonant_non-Fuchsian}
\theta_i \ne 0, \quad \theta_j / \theta_i \notin q^{\mathbb{Z}_{\ge 1}}.
\end{equation}
Let $P(x)=\sum_{n=0}^\infty P_nx^n$ be an $m \times m$ formal power series with $P_0=I$.
We have $Z(qx)=P(qx)^{-1}A(x)P(x)Z(x)$ by the substitution $Y(x)=P(x)Z(x)$.
Set
\begin{equation}
B(x):=P(qx)^{-1}A(x)P(x)=x^r\sum_{n=0}^\infty B_nx^n.
\end{equation}
From the coefficients of $x^0$ in $A(x)P(x)=P(qx)B(x)$,
we have $B_0=A_0$.
Then the coefficients of $x^n \ (n \ge 1)$ gives the following relation:
\begin{align}
B_n=A_n+\sum_{j=1}^{n-1}(A_{n-j} P_j-q^j P_jB_{n-j})+A_0P_n-P_n(q^nA_0).
\end{align}
Set
\begin{equation}
C_n
=
\begin{pmatrix}
C^{(n)}_{11} & C^{(n)}_{12} \\
C^{(n)}_{21} & C^{(n)}_{22}
\end{pmatrix}
:=A_n+\sum_{j=1}^{n-1}(A_{n-j} P_j-q^j P_jB_{n-j})
\end{equation}
for simplicity.
Here $C^{(n)}_{11}$ is $s \times s$, $C^{(n)}_{12}$ is $s \times (m-s)$, $C^{(n)}_{21}$ is $(m-s) \times s$, and $C^{(n)}_{22}$ is $(m-s) \times (m-s)$.
Note that $C_n$ includes $P_1, \ldots, P_{n-1}$.
Then we have
\begin{equation}\label{eq:Pn_equation}
A_0P_n-P_n(q^nA_0)=B_n-C_n.
\end{equation}
Unlike the Fuchsian case, $\mathrm{Sp}(A_0) \cap \mathrm{Sp}(q^nA_0) \ne \emptyset$.
Thus the equation \eqref{eq:Pn_equation} with respect to $P_n$ does not necessarily have a solution.
Instead, we partition $P_n$ conformably with $C_n$
\begin{equation}
P_n=
\begin{pmatrix}
P^{(n)}_{11} & P^{(n)}_{12} \\
P^{(n)}_{21} & P^{(n)}_{22}
\end{pmatrix}
\end{equation}
and choose $P_n$ so that $P^{(n)}_{11}=O, \ P^{(n)}_{22}=O$, and
\begin{align}
B_n
=C_n+
\begin{pmatrix}
O & \Theta P^{(n)}_{12} \\
-P^{(n)}_{21}(q^n\Theta) & O
\end{pmatrix}
\end{align}
to be block-diagonal.
More specifically, we set
\begin{equation}
P^{(n)}_{12}=-\Theta^{-1}C^{(n)}_{12}, \quad
P^{(n)}_{21}=q^{-n}C^{(n)}_{21}\Theta^{-1}.
\end{equation}
Thus we have the following proposition.
\begin{prop}\label{thm:block_diagonalization}
For any system \eqref{eq:original_system_non-Fuchsian} with \eqref{eq:non-Fuchsian_A0} and \eqref{eq:non-resonant_non-Fuchsian} there exists a formal power series with matrix coefficients $P(x)=\sum_{n=0}^\infty P_nx^n$
such that the gauge transformation by $P(x)$ is block-diagonal:
\begin{align}\label{eq:block_diagonal}
Z(qx)=B(x)Z(x), \quad
B(x)=P(qx)^{-1}A(x)P(x)=
\begin{pmatrix}
B_1(x) & O \\
O & B_2(x)
\end{pmatrix}.
\end{align}
The matrix $P(x)$ can be chosen as $P^{(n)}_{11}=O$, $P^{(n)}_{22}=O$.
\end{prop}
Now we apply the above construction to a polynomial coefficient system
\begin{equation}\label{eq:system_polynomial}
Y(qx)=A(x)Y(x), \quad A(x)=A_0+A_1x+\cdots+A_Nx^N.
\end{equation}
If $A_0$ is of the form \eqref{eq:non-Fuchsian_A0}, then the constant term of $B_1(x)$ of \eqref{eq:block_diagonal} is $\Theta$.
Thus the formal normal form of $B_1(x)$ at $x=0$ is $\Theta$.
On the other hand, $B_2(x)$ is of the following form:
\begin{equation}
B_2(x)=x^{r_2}B'_0+x^{r_2+1}B'_1+\cdots
\end{equation}
where $r_2$ is a positive integer.
If $B'_0$ is similar to $\Theta' \oplus O$ where $\Theta'$ is diagonal, invertible, and non-resonant (in particular we assume that $B'_0$ is diagonalizable),
then $B_2(x)$ can be block-diagonalized into the following form
\begin{equation}
\begin{pmatrix}
x^{r_2}\Theta' & O \\
O & x^{r_3}C_2(x)
\end{pmatrix}.
\end{equation}

To summarize the above,
a linear $q$-difference system \eqref{eq:system_polynomial} satisfying diagonalizability (of the first term of each direct summand) and the non-resonant condition can be transformed into the following block diagonal form:
\begin{align}\label{eq:formal_normal_form_0}
Z(qx)=B(x)Z(x), \quad
B(x)=P(qx)^{-1}A(x)P(x)=
\begin{pmatrix}
x^{r_1}B_1 & & \\
 & \ddots & \\
 & & x^{r_k}B_k
\end{pmatrix}
\end{align}
where $B_j \in \mathrm{GL}_{m_j}(\mathbb{C})$.
$r_i$'s are non-negative integers satisfying $r_1=0 < r_2 < \cdots < r_k$.
Here the numbers $r_i$'s and $m_j$'s are uniquely determined only by the original system \eqref{eq:system_polynomial}.
Moreover, if we require that any eigenvalue $\lambda$ of $B_j$ satisfies $|q| < |\lambda| \le 1$, then the conjugacy class of $B_j$ is uniquely determined   (for example see \cite{HSS}).
Then \eqref{eq:formal_normal_form_0} is the formal normal form of \eqref{eq:system_polynomial} at $x=0$.

Similarly, the formal normal form of \eqref{eq:system_polynomial} at $x=\infty$ has the following form:
\begin{equation}\label{eq:formal_normal_form_infinity}
\begin{pmatrix}
x^{N-s_1}B'_1 & & \\
 & \ddots & \\
 & & x^{N-s_\ell}B'_\ell
\end{pmatrix}
\end{equation}
where $B'_j \in \mathrm{GL}_{m'_j}(\mathbb{C})$ and $s_1=0 < s_2 < \cdots < s_\ell$.
The formal normal form at $x=\infty$ is also unique in the same sense as above.

\subsection{Spectral types of linear $q$-difference systems}\label{sec:qST}

First we recall the notion of spectral type of Fuchsian linear $q$-difference systems \cite{SY}.
Consider the following Fuchsian linear $q$-difference system of rank $m$:
\begin{equation}\label{eq:lin_q_diff}
Y(qx)=A(x)Y(x), \quad A(x)=A_0+A_1x+\cdots+A_Nx^N
\end{equation}
where $A_0$ and $A_N$ are invertible.
We assume that, for any $a \in \mathbb{C},\, A(a) \ne O$.
In addition, we assume that $A_0$ and $A_N$ are diagonalizable for simplicity.
Let the eigenvalues of $A_0$ be $\theta_j \ (j=1, \ldots, k)$, and let their multiplicities be $m_j \ (j=1, \ldots, k)$.
Also, let the eigenvalues of $A_N$ be $\kappa_j \ (j=1, \ldots, \ell)$, and let their multiplicities be $n_j \ (j=1, \ldots, \ell)$:
\begin{equation}
A_0 \sim
\theta_1 I_{m_1} \oplus \cdots \oplus \theta_k I_{m_k}, \quad
A_N \sim
\kappa_1 I_{n_1} \oplus \cdots \oplus \kappa_\ell I_{n_\ell}.
\end{equation}
Then we define partitions $S_0$ and $S_\infty$ of $m$ as
\begin{equation}
S_0=m_1, \ldots, m_k, \quad
S_\infty=n_1, \ldots, n_\ell.
\end{equation}

Let $Z_A$ be the set of the zeros of $\det A(x)$:
\begin{equation}
Z_A=\{ a \in \mathbb{C} \,|\, \det A(a)=0 \}=
\{ \alpha_1, \ldots, \alpha_p \}.
\end{equation}
We denote by $d_i \, (i=1, \ldots, m)$ the elementary divisors of $A(x)$.
Here we assume that $d_{i+1} | d_i$.
For any $\alpha_i \in Z_A$, we denote by $\tilde{n}^i_k$ the order of $\alpha_i$ in $d_k$.
Let $\{ n^i_j \}_j$ be the partition conjugate to $\{ \tilde{n}^i_k \}_k$.
Then we define $S_{\mathrm{div}}$ as
\begin{equation}
S_{\mathrm{div}}=n^1_1 \ldots n^1_{k_1}, \ldots, n^p_1 \ldots n^p_{k_p}.
\end{equation}
We call the triple $[S_0; \, S_\infty; \, S_{\mathrm{div}}]$ the \textit{spectral type} of the Fuchsian system \eqref{eq:lin_q_diff}.

Spectral types can also be defined for non-Fuchsian systems.
Taking the formal normal form \eqref{eq:formal_normal_form_0} into account,
we can define $S_0$ as $S_0=\overbrace{( \cdots (\lambda_1)\cdots)}^{\text{$r_1$-tuple}}, \ldots, \overbrace{(\cdots(\lambda_k)\cdots)}^{\text{$r_k$-tuple}}$ where $\lambda_j$ is the partition of $m_j$
determined by the multiplicities of the eigenvalues of $B_j$.
For example, if the normal form of $A(x)$ around $x=0$ is $(B_1) \oplus (x B_2) \oplus (x^3 B_3)$
and the partitions corresponding to $B_j$'s are
\begin{equation}
B_1: 3,1 \quad B_2: 2,1 \quad B_3: 2,2,2
\end{equation}
then 
$S_0=3,1, (2,1), (((2,2,2)))$.

Similarly, taking \eqref{eq:formal_normal_form_infinity} into account,
we define $S_\infty$ as $S_\infty=\overbrace{(\cdots( \lambda'_1 )\cdots)}^{\text{$s_1$-tuple}}, \ldots, \overbrace{(\cdots( \lambda'_\ell )\cdots)}^{\text{$s_\ell$-tuple}}$
where $\lambda'_j$ is the partition of $m'_j$ corresponding to $B'_j$.

$S_{\mathrm{div}}$ is the same as in the Fuchsian case.
Then the triple $[S_0; \, S_\infty; \, S_{\mathrm{div}}]$ is the spectral type.

\section{$q$-matrix $P_\mathrm{VI}$}\label{sec:q-matPVI}
In this section we review the $q$-matrix $P_{\mathrm{VI}}$ \cite{K6},
which describes a connection-preserving deformation of the Fuchsian linear $q$-difference system of spectral type $[m,m; m,m-1,1; m,m,m,m]$
(see \cite{JS, K6} for the connection-preserving deformation).

Consider a linear $q$-difference system of the following form:
\begin{align}\label{eq:linear_qmatPVI}
Y(qx)&=A(x)Y(x), \quad
A(x)=A_0+A_1x+A_2x^2, \quad
A_j \in M_{2m}(\mathbb{C}),
\end{align}
where
\begin{align}
A_2=
\begin{pmatrix}
\kappa_1 I_m & O \\
O & K
\end{pmatrix}, \quad
K=
\mathrm{diag}(\overbrace{\kappa_2, \ldots, \kappa_2}^{m-1}, \kappa_3), \quad
A_0 \sim
\begin{pmatrix}
\theta_1 t I_m & O \\
O & \theta_2 t I_m
\end{pmatrix}.
\end{align}
Since $S_{\mathrm{div}}=m,m,m,m$, 
the Smith normal form of the polynomial matrix $A(x)$ is of the following form:
\begin{equation}
\begin{pmatrix}
I_m & O \\
O & \prod_{i=1}^4 (x-\alpha_i)I_m
\end{pmatrix}.
\end{equation}
We assume that $\alpha_j$'s depend on $t$ as follows:
\begin{equation}
\alpha_j=\left\{
\begin{aligned}
a_j t \quad (j=1,2), \\
a_j \quad (j=3,4).
\end{aligned}
\right.
\end{equation}
We also assume $q\alpha_i \ne \alpha_j \ (i \ne j)$.

The linear $q$-difference systems satisfying the above conditions can be parametrized as follows:
\begin{align}
A(x)&=
\begin{pmatrix}
WK\{ \kappa_1(xI_m-F)(xI_m-\bm{\alpha})+\kappa_1G_1 \}K^{-1}W^{-1} & WK(xI_m-F) \\
\kappa_1(\bm{\gamma}x+\bm{\delta})W^{-1} & K(xI_m-\bm{\beta})(xI_m-F)+KG_2
\end{pmatrix}
\end{align}
where
\begin{align}
\bm{\alpha}&=(\kappa_1-K)^{-1}\left\{ (\theta_1+\theta_2)t F^{-1}-\kappa_1F^{-1}G_1-KG_2F^{-1}+K(F+G_1^{-1}FG_1+\beta_1) \right\}, \\
\bm{\beta}&=(\kappa_1-K)^{-1}\left\{ -(\theta_1+\theta_2)t F^{-1}+\kappa_1F^{-1}G_1+KG_2F^{-1}-\kappa_1(F+G_1^{-1}FG_1+\beta_1) \right\}, \\
\bm{\gamma}&=K\{ G_1+G_2+F\bm{\alpha}+\bm{\beta}F+\bm{\beta}\bm{\alpha}-G_1^{-1}(F^2+\beta_1F+\beta_2)G_1 \}K^{-1}, \\
\bm{\delta}&=\kappa_1^{-1}\{ t^2\theta_1\theta_2 F^{-1}-\kappa_1 K(G_2+\bm{\beta}F)F^{-1}(G_1+F\bm{\alpha}) \}K^{-1}.
\end{align}
Here the auxiliary parameters $\beta_j$'s are defined by
\begin{equation}
\sum_{j=0}^4 \beta_{4-j}z^j:=\prod_{j=1}^4(z-\alpha_j).
\end{equation}
The matrices $G_1$ and $G_2$ satisfy
\begin{align}
G_1G_2&=(F-\alpha_1I)(F-\alpha_2I)(F-\alpha_3I)(F-\alpha_4I) \label{eq:G1G2rel}.
\end{align}
The relation \eqref{eq:G1G2rel} allows us to introduce a new variable $G$ by
\begin{align}
G_1=q^{-1}\kappa_1^{-1}(F-\alpha_1)(F-\alpha_2)G^{-1}, \quad
G_2=q\kappa_1G(F-\alpha_3)(F-\alpha_4).
\end{align}
Then $F$ and $G$ satisfy the following commutation relation:
\begin{align}\label{eq:commutation_relation}
F^{-1}GFG^{-1}&=\rho K, \quad \rho=\frac{a_1a_2a_3a_4\kappa_1}{\theta_1\theta_2}. 
\end{align}

Since
\begin{equation}
\det A(x)=\kappa_1^m \kappa_2^{m-1}\kappa_3 \prod_{i=1}^4(x-\alpha_i)^m,
\end{equation}
we have
\begin{equation}\label{eq:qFuchsrel}
{\kappa_1}^m {\kappa_2}^{m-1}\kappa_3 \prod_{i=1}^4 {a_i}^m={\theta_1}^m{\theta_2}^m.
\end{equation}

Let us consider the connection-preserving deformation of the system \eqref{eq:linear_qmatPVI}.
We choose $t$ as a deformation parameter.
The parameters $\theta_j$, $\kappa_j$, and $a_j$'s are independent of $t$.
In the following we write $A(x,t)$ instead of $A(x)$ when it is necessary to emphasize that $A(x)$ depends on $t$.

The connection-preserving deformation of \eqref{eq:linear_qmatPVI} is given by
\begin{align}\label{eq:B_qmatPVI}
Y(x, qt)=B(x, t)Y(x, t)
\end{align}
where
\begin{align}
B(x,t)=
\frac{x(xI+B_0)}{(x-qa_1t)(x-qa_2t)}, \quad
B_0=
\begin{pmatrix}
B_{11} & B_{12} \\
B_{21} & B_{22}
\end{pmatrix}.
\end{align}
Here $B_{ij}$'s are $m \times m$ matrices and given as follows:
\begin{align}
B_{11}&=
qWK(I-\overline{G}K)^{-1}\overline{G}K \left[ K^{-1}\overline{G}^{-1}\{ F-(a_1+a_2)t \}+\bm{\beta} \right]K^{-1}W^{-1}, \\
B_{12}&=
qWK(I-\overline{G}K)^{-1}\overline{G}, \\
B_{21}&=q\kappa_1\left\{ q^{-1}\kappa_1^{-1}(\overline{F}-qa_2t)\overline{G}^{-1}-qa_1t+\overline{\bm{\alpha}} \right\}
(I-q\kappa_1\overline{G})^{-1} \nonumber \\
&\quad \times \overline{G}K
\left\{ K^{-1}\overline{G}^{-1}( F-a_2t )-a_1t+\bm{\beta} \right\}K^{-1}W^{-1} \\
&=q\kappa_1\left\{ q^{-1}\kappa_1^{-1}(\overline{F}-qa_1t)\overline{G}^{-1}-qa_2t+\overline{\bm{\alpha}} \right\}
(I-q\kappa_1\overline{G})^{-1} \nonumber \\
&\quad \times \overline{G}K
\left\{ K^{-1}\overline{G}^{-1}( F-a_1t )-a_2t+\bm{\beta} \right\}K^{-1}W^{-1}, \\
B_{22}&=\left[ q^{-1}\kappa_1^{-1} \{ \overline{F}-q(a_1+a_2)t \}\overline{G}^{-1}+\overline{\bm{\alpha}} \right]
q\kappa_1\overline{G}(I-q\kappa_1\overline{G})^{-1}.
\end{align}
Here the overline denotes the $q$-shift with respect to $t$: $\overline{f}=f(qt)$ for $f=f(t)$.

Now we have the pair of linear $q$-difference systems:
\begin{equation}\label{eq:qLax}
\left\{
\begin{aligned}
Y(qx,t)=A(x,t)Y(x,t), \\
Y(x,qt)=B(x,t)Y(x,t).
\end{aligned}
\right.
\end{equation}
Then the compatibility condition of \eqref{eq:qLax}
\begin{equation}
A(x, qt)B(x, t)=B(qx, t)A(x, t)
\end{equation}
reduces to a system of non-linear $q$-difference equations satisfied by $F$, $G$, and $W$.

\begin{thm}{\cite{K6}}
The compatibility condition $A(x,qt)B(x,t)=B(qx,t)A(x,t)$ is equivalent to
\begin{align}
\overline{G}KG&=\frac{1}{q\kappa_1}(F-a_1t)(F-a_2t)(F-a_3)^{-1}(F-a_4)^{-1}, \label{eq:Geq_thm}\\
\overline{F}KF&=
\frac{\theta_1\theta_2}{\kappa_1a_1a_2}\left( \overline{G}-t\frac{a_1a_2}{\theta_1} \right)\left( \overline{G}-t\frac{a_1a_2}{\theta_2} \right)
\left( \overline{G}-\frac{1}{q\kappa_1} \right)^{-1}\left( \overline{G}-\rho \right)^{-1}, \label{eq:Feq_thm}\\
W^{-1}\overline{W}&=q\kappa_1(\overline{G}-K^{-1})^{-1}\left( \overline{G}-\frac{1}{q\kappa_1} \right)K^{-1}. \label{eq:Weq_thm}
\end{align}
\end{thm}
We call the system \eqref{eq:Geq_thm} and \eqref{eq:Feq_thm} (with \eqref{eq:commutation_relation}) the $q$-matrix $P_{\mathrm{VI}}$.
Although this system appears to have eight parameters ($\theta_i$'s, $\kappa_i$'s, and $a_i$'s with a single relation),
the number of parameters can be reduced to five by rescaling $F$, $G$, and $t$.

\section{Degeneration of $q$-matrix $P_\mathrm{VI}$}\label{sec:degeneration}
Now we consider a degeneration of the $q$-matrix $P_{\mathrm{VI}}$ which corresponds to the limit $\kappa_1$ to 0.

\subsection{From $q$-matrix $P_\mathrm{VI}$ to $q$-matrix $P_\mathrm{V}$}
Consider the following transformation:
\begin{equation}\label{eq:transf_degeneration}
\begin{aligned}
&t=\varepsilon \tilde{t}, \quad F=\varepsilon \tilde{F}, \quad G=\varepsilon \tilde{G}, \quad W=\varepsilon \tilde{W}, \\
&a_3=-\varepsilon \tilde{a}_3, \quad a_4=-\varepsilon^{-1} \tilde{\kappa}_1, \quad \kappa_1=\varepsilon, \quad \kappa_2=\varepsilon^{-1}\tilde{\kappa}_2, \quad
\kappa_3=\varepsilon^{-1}\tilde{\kappa}_3.
\end{aligned}
\end{equation}
We set $\tilde{K}=\mathrm{diag}(\overbrace{\tilde{\kappa}_2, \ldots, \tilde{\kappa}_2}^{m-1}, \tilde{\kappa}_3)$
so that we have $K=\varepsilon^{-1}\tilde{K}$.
The other parameters $a_1, a_2$ and $\theta_1, \theta_2$ are not changed.
This transformation is compatible with the commutation relation~\eqref{eq:commutation_relation},
that is, $\tilde{F}^{-1}\tilde{G}\tilde{F}\tilde{G}^{-1}=\tilde{\rho}\tilde{K}$ holds
where $\tilde{\rho}=\frac{a_1 a_2 \tilde{a}_3 \tilde{\kappa}_1}{\theta_1\theta_2}$.

Substituting \eqref{eq:transf_degeneration} into \eqref{eq:Geq_thm}, we have
\begin{equation}
\varepsilon \overline{\tilde{G}}\tilde{K}\tilde{G}=
\frac{\varepsilon}{q}(\tilde{F}-a_1\tilde{t})(\tilde{F}-a_2\tilde{t})(\tilde{F}+\tilde{a}_3)^{-1}(\varepsilon^2\tilde{F}+\tilde{\kappa}_1)^{-1}.
\end{equation}
Letting $\varepsilon \to 0$, we obtain
\begin{equation}
\overline{\tilde{G}}\tilde{K}\tilde{G}=
\frac{1}{q\tilde{\kappa}_1}(\tilde{F}-a_1\tilde{t})(\tilde{F}-a_2\tilde{t})(\tilde{F}+\tilde{a}_3)^{-1}.
\end{equation}
Similarly, from the equation \eqref{eq:Feq_thm} we have
\begin{equation}
\varepsilon\overline{\tilde{F}}\tilde{K}\tilde{F}=
\varepsilon\frac{\theta_1\theta_2}{a_1a_2}
\left( \overline{\tilde{G}}-\tilde{t}\frac{a_1a_2}{\theta_1} \right)\left( \overline{\tilde{G}}-\tilde{t}\frac{a_1a_2}{\theta_2} \right)
\left( \varepsilon^2\overline{\tilde{G}}-\frac{1}{q} \right)^{-1}\left( \overline{\tilde{G}}-\tilde{\rho} \right)^{-1}.
\end{equation}
Letting $\varepsilon \to 0$, we obtain
\begin{equation}
\overline{\tilde{F}}\tilde{K}\tilde{F}=
-q\frac{\theta_1\theta_2}{a_1a_2}
\left( \overline{\tilde{G}}-\tilde{t}\frac{a_1a_2}{\theta_1} \right)\left( \overline{\tilde{G}}-\tilde{t}\frac{a_1a_2}{\theta_2} \right)
\left( \overline{\tilde{G}}-\tilde{\rho} \right)^{-1}.
\end{equation}
From the equation \eqref{eq:Weq_thm} we have
\begin{equation}
{\tilde{W}}^{-1}\overline{\tilde{W}}=q(\overline{\tilde{G}}-\tilde{K}^{-1})^{-1}\left( \varepsilon^2\overline{\tilde{G}}-\frac{1}{q} \right)\tilde{K}^{-1}.
\end{equation}
Letting $\varepsilon \to 0$, we obtain
\begin{equation}
{\tilde{W}}^{-1}\overline{\tilde{W}}=-(\tilde{K}\overline{\tilde{G}}-I)^{-1}.
\end{equation}

Omitting the tilde, we obtain the following system of non-linear $q$-difference equations
\begin{align}
\overline{G}KG&=\frac{1}{q \kappa_1}(F-a_1t)(F-a_2t)(F+a_3)^{-1}, \\
\overline{F}KF&=
-q\frac{\theta_1\theta_2}{a_1a_2}
\left( \overline{G}-t\frac{a_1a_2}{\theta_1} \right)\left( \overline{G}-t\frac{a_1a_2}{\theta_2} \right)
\left( \overline{G}-\rho \right)^{-1}, \\
W^{-1}\overline{W}&=(I-K\overline{G})^{-1}.
\end{align}

The associated linear system \eqref{eq:linear_qmatPVI} can also be degenerated in the same manner as above.
Set $x=\varepsilon \tilde{x}$.
Notice that
\begin{equation}
\bm{\alpha}=-\tilde{\kappa}_1 \varepsilon^{-1}+O(1), \quad
\bm{\beta}=O(\varepsilon), \quad
\bm{\gamma}=O(1), \quad
\bm{\delta}=O(\varepsilon), \quad
G_1=O(1), \quad
G_2=O(\varepsilon^2).
\end{equation}
We set
\begin{equation}
\tilde{\bm{\beta}}:=\lim_{\varepsilon \to 0} \frac{\bm{\beta}}{\varepsilon}, \quad
\tilde{\bm{\gamma}}:=\lim_{\varepsilon \to 0} \bm{\gamma}, \quad
\tilde{\bm{\delta}}:=\lim_{\varepsilon \to 0} \frac{\bm{\delta}}{\varepsilon}, \quad
\tilde{G}_1:=\lim_{\varepsilon \to 0} G_1, \quad
\tilde{G}_2:=\lim_{\varepsilon \to 0} \frac{G_2}{\varepsilon^2}.
\end{equation}
Then it is easy to see that
\begin{align}
\tilde{A}(\tilde{x}):=
\lim_{\varepsilon \to 0}\varepsilon^{-1}A(x)=
\begin{pmatrix}
\tilde{W}\tilde{K}\{ \tilde{\kappa}_1(\tilde{x}I_m-\tilde{F})+\tilde{G}_1 \}\tilde{K}^{-1}\tilde{W}^{-1} & \tilde{W}\tilde{K}(\tilde{x}I_m-\tilde{F}) \\
(\tilde{\bm{\gamma}}\tilde{x}+\tilde{\bm{\delta}})\tilde{W}^{-1} & \tilde{K}(\tilde{x}I_m-\tilde{\bm{\beta}})(\tilde{x}I_m-\tilde{F})+\tilde{K}\tilde{G}_2
\end{pmatrix}.
\end{align}

\begin{rem}
The multiplication of $A(x)$ by $\varepsilon^{-1}$ can be realized by a simple gauge transformation of the linear system.
For example, consider the transformation $Y=x^{\log\varepsilon / \log q} \tilde{Y}$
or use the ratio of theta functions \eqref{eq:theta} instead of $x^{\log\varepsilon / \log q}$.
Then we have $\tilde{Y}(qx)=\varepsilon^{-1}A(x)\tilde{Y}(x)$.
\end{rem}
Thus we obtain (by omitting the tilde)
\begin{align}\label{eq:parametrization_of_LS}
A(x)&=
\begin{pmatrix}
WK\{ \kappa_1(xI_m-F)+G_1 \}K^{-1}W^{-1} & WK(xI_m-F) \\
(\bm{\gamma}x+\bm{\delta})W^{-1} & K(xI_m-\bm{\beta})(xI_m-F)+K G_2
\end{pmatrix} \nonumber \\
&=:A_0+A_1x+A_2x^2,
\end{align}
where
\begin{align}
\bm{\beta}&=K^{-1}\{ (\theta_1+\theta_2)t F^{-1}-F^{-1}G_1-KG_2F^{-1}+\kappa_1 \}, \\
\bm{\gamma}&=K\{ G_1-\kappa_1(F+\bm{\beta}+GFG^{-1}-(a_1+a_2)t+a_3) \}K^{-1}, \\
\bm{\delta}&=F^{-1}(G_1-\kappa_1 F-\theta_1 t)(G_1-\kappa_1 F-\theta_2 t)K^{-1}.
\end{align}
From the determinant of \eqref{eq:parametrization_of_LS}
we have
\begin{equation}
{\kappa_1}^m {\kappa_2}^{m-1}\kappa_3 \prod_{i=1}^3 {a_i}^m={\theta_1}^m{\theta_2}^m.
\end{equation}
The matrices $G_1$ and $G_2$ are given by
\begin{equation}
G_1=q^{-1}(F-a_1 t)(F-a_2 t)G^{-1}, \quad
G_2=q\kappa_1 G(F+a_3)
\end{equation}
and satisfy
\begin{equation}
G_1G_2=\kappa_1(F-a_1t)(F-a_2t)(F+a_3).
\end{equation}
The matrices $F$ and $G$ satisfy the following commutation relation:
\begin{equation}\label{eq:commutation_relation_qmatPV}
F^{-1}GFG^{-1}=\rho K, \quad \rho=\frac{a_1a_2a_3\kappa_1}{\theta_1\theta_2}. 
\end{equation}

The system in $t$-direction \eqref{eq:B_qmatPVI} can also be degenerated in the same manner.
As a result, we have (by omitting the tilde)
\begin{align}\label{eq:B_qmatPV}
B(x,t)=
\frac{x(xI+B_0)}{(x-qa_1t)(x-qa_2t)}, \quad
B_0=
\begin{pmatrix}
B_{11} & B_{12} \\
B_{21} & B_{22}
\end{pmatrix}
\end{align}
where $B_{ij}$'s are $m \times m$ matrices given by
\begin{align}
B_{11}&=
qWK(I-\overline{G}K)^{-1} \left\{ F-(a_1+a_2)t+\overline{G}K\bm{\beta} \right\}K^{-1}W^{-1}, \\
B_{12}&=
qWK(I-\overline{G}K)^{-1}\overline{G}, \\
B_{21}&=\left\{ (\overline{F}-qa_2t)\overline{G}^{-1}-q\kappa_1 \right\}
\left( F-a_2t-a_1t\overline{G}K+\overline{G}K\bm{\beta} \right)K^{-1}W^{-1} \\
&=\left\{ (\overline{F}-qa_1t)\overline{G}^{-1}-q\kappa_1 \right\}
\left( F-a_1t-a_2t\overline{G}K+\overline{G}K\bm{\beta} \right)K^{-1}W^{-1}, \\
B_{22}&=\overline{F}-q(a_1+a_2)t-q\kappa_1\overline{G}.
\end{align}
We have the following theorem.
\begin{thm}
The compatibility condition $A(x,qt)B(x,t)=B(qx,t)A(x,t)$ with \eqref{eq:parametrization_of_LS} and \eqref{eq:B_qmatPV} is equivalent to
\begin{align}
\overline{G}KG&=
\frac{1}{q \kappa_1}(F-a_1 t)(F-a_2 t)(F+a_3)^{-1}, \label{eq:q-matPVG} \\
\overline{F}KF&=
-q\frac{\theta_1\theta_2}{a_1a_2}
\left( \overline{G}-t\frac{a_1a_2}{\theta_1} \right)\left( \overline{G}-t\frac{a_1a_2}{\theta_2} \right)
\left( \overline{G}-\rho \right)^{-1}, \label{eq:q-matPVF} \\
W^{-1}\overline{W}&=-(K\overline{G}-I)^{-1}. \label{eq:q-matPVW}
\end{align}
\end{thm}
\begin{proof}
Direct calculation.
\end{proof}
We call the system \eqref{eq:q-matPVG} and \eqref{eq:q-matPVF} (with \eqref{eq:commutation_relation_qmatPV}) the $q$-matrix fifth Painlev\'e system
($q$-matrix $P_{\mathrm{V}}$).
Although this system appears to have seven parameters ($\theta_i$'s, $\kappa_i$'s, and $a_i$'s with a single relation),
the number of parameters can be reduced to four by rescaling $F$, $G$, and $t$.

\subsection{Characterization of the linear system}
The matrix \eqref{eq:parametrization_of_LS} satisfies

\noindent
(C1): $A_0$ is similar to $\theta_1 t I_m \oplus \theta_2 t I_m$.

\noindent
(C2): The formal normal form of $A(x)$ at $x=\infty$ is
\begin{equation}
\begin{pmatrix}
x^2 K & O \\
O & x (\kappa_1 I_m)
\end{pmatrix}.
\end{equation}

\noindent
(C3): The Smith normal form of $A(x)$ is
\begin{equation}
\begin{pmatrix}
I_m & O \\
O & \prod_{j=1}^3(x-\alpha_j)I_m
\end{pmatrix}.
\end{equation}
Conversely, it can be shown that a polynomial matrix $A(x)$ satisfying the above three conditions can be written (generically) in the form \eqref{eq:parametrization_of_LS}.
Thus the linear system associated with the $q$-matrix $P_{\mathrm{V}}$
is characterized by the conditions (C1), (C2), and (C3).
From the definition given in Section~\ref{sec:qST},
the spectral type of the system is written as $[m,m;\, m-1, 1,(m);\, m,m,m]$.

\section{Continuous limit of $q$-matrix $P_{\mathrm{V}}$}\label{sec:CL}

The system \eqref{eq:q-matPVG} and \eqref{eq:q-matPVF} can be viewed as a $q$-analogue of the matrix $P_{\mathrm{V}}$ \eqref{eq:non-abelian_mPV}.
That is, taking the limit $q \to 1$, one can obtain \eqref{eq:non-abelian_mPV} from \eqref{eq:q-matPVG} and \eqref{eq:q-matPVF}. 
In fact, let us define the parameter $\varepsilon$ by $q=1+\varepsilon$.
We set
\begin{equation}
\begin{aligned}
&\theta_i=1+\sigma_i \varepsilon \ (i=1,2), \quad 
\kappa_1=-1+\mu_1\varepsilon, \quad
\kappa_i=-\varepsilon(1-\mu_i \varepsilon) \ (i=2,3), \\
&a_i=1-\zeta_i \varepsilon \ (i=1, 2), \quad a_3=-\varepsilon^{-1},
\end{aligned}
\end{equation}
and $M=\mathrm{diag}(\overbrace{\mu_2, \ldots, \mu_2}^{m-1}, \mu_3)$.
Moreover, we introduce new dependent variables $Q$ and $P$
which are related to $F$ and $G$ by
\begin{align}
&F=-(\tilde{P}-\varepsilon^{-1}t)(\phi_1-\phi_2\tilde{Q})^{-1}\tilde{Q}, \quad
G=(\tilde{P}-\varepsilon^{-1}t)(\phi_1-\phi_2\tilde{Q})^{-1}, \\
&\phi_1=-\varepsilon^{-1}-1-\zeta_1-\zeta_2-\frac{\sigma_1+\sigma_2}{2}, \quad
\phi_2=-\varepsilon^{-1}-\frac{\zeta_1+\zeta_2}{2}, \\
&\tilde{Q}=I-\hat{Q}^{-1}, \quad
\tilde{P}=t\left\{ (\hat{Q}-I)\hat{P}\hat{Q}+\frac{\zeta_2-\zeta_1+\sigma_1-\sigma_2}{2}\hat{Q}+\frac{\zeta_1-\zeta_2}{2} \right\}, \\
&\hat{Q}=g^{-1}Qg, \quad \hat{P}=g^{-1}Pg.
\end{align}
Here $g=t^M$,
which is a solution to $\frac{dg}{dt}g^{-1}=\frac{1}{t}M$.

Then, taking the limit $\varepsilon \to 0$, we find that $Q$ and $P$ satisfy the following equations:
\begin{align}
&t\frac{dQ}{dt}=Q(Q-1)(P+t)+P(Q-1)Q-(\zeta_1-\zeta_2)(Q-1)+(\sigma_1-\sigma_2)Q, \\
&t\frac{dP}{dt}=-(Q-1)P(P+t)-(P+t)PQ-(\zeta_2-\zeta_1+\sigma_1-\sigma_2)P-(\zeta_2+\zeta_4+\sigma_1)t.
\end{align}
These equations coincide with \eqref{eq:non-abelian_mPV} by the following correspondence of the parameters:
\begin{align}
&\sigma_1-\sigma_2 = \theta^0, \quad \zeta_1-\zeta_2 = \theta^1, \quad 
\mu_i+\zeta_2+\sigma_2 = \theta^\infty_i \ (i=1,2,3).
\end{align}

Expanding \eqref{eq:commutation_relation_qmatPV} with respect to the small parameter $\varepsilon$ and
taking the coefficient of $\varepsilon^1$, we have the commutation relation between $P$ and $Q$:
\begin{equation}
PQ-QP=(\mu_1+\zeta_1+\zeta_2+\sigma_1+\sigma_2)I_m+M.
\end{equation}

The linear system \eqref{eq:parametrization_of_LS} also admits the continuous limit in a similar way.
To see this, we first change the dependent variable $Y$ to $Z$:
$Y(x)=f(x)Z(x)$, where $f(x)$ is a solution of the following $q$-difference equation
\begin{equation}
f(qx)=-(x-t)f(x).
\end{equation}
For example, we can take
\begin{equation}
f(x)=\frac{\vartheta_q(x/t)}{(x/t; q)_\infty\vartheta_q(x)},
\end{equation}
where
\begin{equation}\label{eq:theta}
(a; q)_\infty =\prod_{n=0}^{\infty}(1-aq^n), \quad
\vartheta_q(x)=\prod_{n=0}^\infty (1-q^{n+1})(1+xq^n)(1+x^{-1}q^{n+1}).
\end{equation}
Then we have
\begin{align}
\frac{Z(x)-Z(qx)}{(1-q)x}=\frac{1}{-\varepsilon x}\left\{ I_{2m}-\frac{1}{-(x-t)}A(x) \right\}Z(x).
\end{align}
Set $W=t U^{-1}$.
Define matrices $\mathsf{A}_0$, $\mathsf{A}_1$, and $\mathsf{A}_\infty$ by
\begin{equation}
\lim_{\varepsilon \to 0}\frac{1}{-\varepsilon x}\left\{ I_{2m}-\frac{1}{-(x-t)}A(x) \right\}=
\frac{\mathsf{A}_0}{x}+\frac{\mathsf{A}_1}{x-t}+\mathsf{A}_\infty.
\end{equation}
It can be shown that  the matrices $\mathsf{A}_0$, $\mathsf{A}_1$, and $\mathsf{A}_\infty$ (almost) coincide with \eqref{eq:matrices_matrixPV}.
More precisely, performing suitable scalar gauge transformations
(in other words, adding suitable scalar matrices to $\mathsf{A}_0$, $\mathsf{A}_1$, and $\mathsf{A}_\infty$),
performing the gauge transformation by
\begin{equation}
h=
\begin{pmatrix}
O & I_m \\
I_m & O
\end{pmatrix},
\end{equation}
and setting $x=t \tilde{x}$, we have
\begin{align}
h^{-1}(\mathsf{A}_0-\sigma_2 I_{2m})h=\mathcal{A}_0, \quad
h^{-1}(\mathsf{A}_1-\zeta_2 I_{2m})h=\mathcal{A}_1, \quad
h^{-1}(t\mathsf{A}_\infty-t I_{2m})h=\mathcal{A}_\infty.
\end{align}
Thus the resulting system of linear differential equations
\begin{align}
\frac{d\tilde{Z}}{d\tilde{x}}=\left( \frac{\mathcal{A}_0}{\tilde{x}}+\frac{\mathcal{A}_1}{\tilde{x}-1}+\mathcal{A}_\infty \right)\tilde{Z}
\end{align}
coincides with the $x$-direction of \eqref{eq:Lax_matPV}.

\appendix
\section{The matrix fifth Painlev\'e system}\label{sec:mpVI}

In this appendix, we review the matrix fifth Painlev\'e system (matrix $P_{\mathrm{V}}$)~\cite{K2, K5}.
The matrix $P_{\mathrm{V}}$ is derived from the isomonodromic deformation of a certain linear differential system.
There are several Lax pairs for the matrix $P_\mathrm{V}$, one of them is the following:
\begin{equation}\label{eq:Lax_matPV}
\left\{
\begin{aligned}
\frac{\partial Y}{\partial x}&=
\left(
\frac{\mathcal{A}_0}{x}+\frac{\mathcal{A}_1}{x-1}+\mathcal{A}_\infty
\right)Y, \\
\frac{\partial Y}{\partial t}&=(-E_2\otimes I_m x+\mathcal{B}_1)Y ,
\end{aligned}
\right.
\end{equation}
where
\begin{equation}\label{eq:matrices_matrixPV}
\begin{aligned}
\mathcal{A}_{\xi}&=
(I_m \oplus U)^{-1}\hat{\mathcal{A}}_{\xi}(I_m \oplus U) \quad (\xi=0, 1, \infty),\\
\mathcal{A}_\infty&=
\begin{pmatrix}
O_m & O_m \\
O_m & -t I_m
\end{pmatrix},\quad
\hat{\mathcal{A}}_0=
\begin{pmatrix}
QP+\theta^0+\theta^\infty_1 \\
tI_m
\end{pmatrix}
\left(I_m-Q,\ \frac{1}{t}\{(Q-I_m)QP+(\theta^0+\theta^\infty_1)Q-\theta^\infty_1\}\right),\\
\hat{\mathcal{A}}_1&=
\begin{pmatrix}
(Q-I_m)PQ+(\theta^0+\theta^\infty_1)Q+\theta^1 \\
tQ
\end{pmatrix}
\left(I_m,\ \frac{1}{t}
\{(I_m-Q)P-\theta^0-\theta^\infty_1\}\right),\\
E_2&=
\begin{pmatrix}
0 & 0 \\
0 & 1
\end{pmatrix}, \quad
\Theta=
\begin{pmatrix}
\theta^\infty_2 I_{m-1} & \\
 & \theta^\infty_3
\end{pmatrix}.
\end{aligned}
\end{equation}
Furthermore, the matrix $\mathcal{B}_1$ is given by
\begin{align}
\mathcal{B}_1&=
(I_m \oplus U)^{-1}
\begin{pmatrix}
O_m & \frac{[\hat{\mathcal{A}}_0+\hat{\mathcal{A}}_1]_{12}}{t}\\
\frac{[\hat{\mathcal{A}}_0+\hat{\mathcal{A}}_1]_{21}}{t} & O_m
\end{pmatrix}
(I_m \oplus U),
\end{align}
where $[\hat{\mathcal{A}}_0+\hat{\mathcal{A}}_1]_{ij}$
is the $(i,j)$-block of the matrix $\hat{\mathcal{A}}_0+\hat{\mathcal{A}}_1$.
The Fuchs-Hukuhara relation is written as
$m(\theta^0 +\theta^1 +\theta_1^\infty) +(m-1)\theta_2^\infty+\theta_3^\infty =0$.
$P$ and $Q$ satisfy $[P, Q]=(\theta^0+\theta^1+\theta^\infty_1)I_m+\Theta$.
The system in $x$-direction of \eqref{eq:Lax_matPV} is characterized by the spectral type $(m)(m-1\, 1),mm,mm$.

The compatibility condition (in other words, isomonodromic deformation equation) for \eqref{eq:Lax_matPV} has two descriptions,
which are mutually equivalent.
One is the Hamiltonian form and the other is the ``non-abelian'' form.
The Hamiltonian is given by
\begin{equation}
tH^{\mathrm{Mat}, m}_{\mathrm{V}}
\left({-\theta^0-\theta^1-\theta^\infty_1,\theta^0-\theta^1 \atop \theta^1, \theta^0+\theta^1+\theta^\infty_1+\theta^\infty_2};
t; Q, P \right)=
\mathrm{tr}[P(P+t)Q(Q-1)+(\theta^0-\theta^1)PQ+\theta^1 P+(\theta^0+\theta^\infty_1)tQ].
\end{equation}
Then the compatibility condition can be written as follows:
\begin{align}
\frac{dq_{ij}}{dt}=\frac{\partial H^{\mathrm{Mat}, m}_{\mathrm{V}}}{\partial p_{ji}}, \quad
\frac{dp_{ij}}{dt}=-\frac{\partial H^{\mathrm{Mat}, m}_{\mathrm{V}}}{\partial q_{ji}}.
\end{align}

On the other hand, the non-abelian description is given as follows~\cite{K5}:
\begin{equation}\label{eq:non-abelian_mPV}
\left\{
\begin{aligned}
t\frac{dQ}{dt}&=Q(Q-1)(P+t)+PQ(Q-1)+(\theta^0-\theta^1)Q+\theta^1,\\
t\frac{dP}{dt}&=-(Q-1)P(P+t)-P(P+t)Q-(\theta^0-\theta^1)P-(\theta^0+\theta^\infty_1)t.
\end{aligned}
\right.
\end{equation}

\end{document}